\RequirePackage[l2tabu,orthodox]{nag}
\pdfoutput=1
\documentclass[1p]{elsarticle}
\usepackage{amsmath,amssymb,amsthm}
\usepackage{color}
\usepackage{algpseudocode}
\usepackage{setspace}
\usepackage{graphicx}
\usepackage{relsize}
\usepackage{lscape}
\usepackage{arydshln}
\usepackage{framed}

\usepackage{tikz}
\usetikzlibrary{arrows}

\usepackage{natbib}
\definecolor{darkgreen}{rgb}{0.0,0.7,0.0}

\definecolor{darkred}{rgb}{0.75,0.0,0.0}

\newcommand{\R}{\mathbb{R}}
\DeclareMathOperator{\var}{Var}

\newcommand{\XX}{\mathcal{X}}

\theoremstyle{plain}
\newtheorem{thm}{Theorem}[section]
\newtheorem{lem}[thm]{Lemma}

\theoremstyle{definition}
\newtheorem{defn}{Definition}[section]

\newtheorem{rem}{Remark}[section]

\begin{document}

\begin{frontmatter}

\title{Systematic evaluation of the population-level effects of alternative treatment strategies on the basic reproduction number}

\author[DG]{Dmitry Gromov\corref{corrauth}}
\cortext[corrauth]{Corresponding author}
\ead{dv.gromov@gmail.com}

\author[IB1,IB2]{Ingo Bulla}\ead{ingobulla@gmail.com}
\author[ERS]{Ethan O. Romero-Severson}\ead{eoromero@lanl.gov}

\address[DG]{Faculty of Applied Mathematics and Control Processes, Saint Petersburg State University,\\ St. Petersburg, Russia}
\address[IB1]{Institut f\"ur Mathematik und Informatik, Universit\"at Greifswald, Walther-Rathenau-Stra\ss e 47, 17487 Greifswald, Germany}
\address[IB2]{Universit\'e Perpignan Via Domitia, IHPE UMR 5244, CNRS, F-66860 Perpignan, France}
\address[ERS]{Theoretical Biology and Biophysics Group, Los Alamos National Laboratory, \\Los Alamos, New Mexico, USA}

\begin{abstract}

An approach to estimate the influence of the treatment-type controls on the basic reproduction number, $R_0$, is proposed and elaborated. The presented approach allows one to estimate the effect of a given treatment strategy or to compare a number of different treatment strategies on the basic reproduction number. All our results are valid for sufficiently small values of the control. However, in many cases it is possible to extend this analysis to larger values of the control as was illustrated by examples. 
\end{abstract}


\end{frontmatter}


\section{Introduction}
The basic reproduction number ($R_0$), i.e. the number of infections generated by an infected person in a fully susceptible population, is a key determinant of the dynamics of an infectious disease. 
Interpretation of $R_0$ is complex and is covered in a number of excellent papers, see, e.g., \cite{Driessche:02,Heesterbeek:02,Heffernan:05} as well as the books \cite{Diekmann:00,Diekmann:12}. 
For the purposes of this paper, we note that $R_0$ can be interpreted as a threshold parameter: if $R_0<1$ then the disease can not sustain itself in the population and will eventually die out.
That is, if our goal is to push a disease towards extinction, then we can categorize a population-level intervention as `effective' if it reduces $R_0$ from where it was before the intervention was enacted.

\paragraph{Computation of $R_0$}  
$R_0$ can be computed or estimated from partially empirical considerations, for instance the ``survival function''-based approach \cite{Heesterbeek:96} that assumes knowledge of both the survival probability $F(t)$ and the infectivity $b(t)$ of an individual as functions of time. In addition to being difficult to implement, such approaches can hardly be used for computing $R_0$ in the case where there are more than one group of infected. This issue was addressed in \cite{Diekmann:90}, where an elegant approach to computing $R_0$ was proposed, termed the {\em next generation matrix} (NGM) method. This method boils down to computing $R_0$ as the spectral radius of a specially constructed matrix. Later, this method was detailed and substantiated in a number of papers, see, e.g., \cite{Driessche:02} and \cite{Castillo:02}. The power of the NGM approach lies in its universality. It can be applied to any population balance model as long as it satisfies a number of very natural assumptions. The result is based on certain properties of positive and inverse positive matrices, in particular M-matrices \cite{BP:94,Horn:91}. 

\paragraph{Contribution}
We aim at developing this approach in that we extend the notion of the basic reproduction number to a class of controlled disease propagation models. Much of applied theoretical epidemiology focuses on how  treatment policies, intervention designs, and novel treatments will impact the burden of a given disease by considering alternative scenarios. 
There have been a large number of papers aimed at evaluating the efficacy of different treatment schemes using various methods, from numerical simulation to optimal control. We mention \cite{Castillo:97,Bowden:00,Korenromp:00,blower_forecasting_2003,Baggaley:05,Sharomi:08, Martin:13,SilvaTorres:15,MMB:17} for a short list of related research. Modeling efforts have even attempted to measure the effects of control programs in historical epidemics \cite{bootsma_effect_2007}. However, most modeling efforts do not attempt to measure intervention effects systematically, i.e. by considering the joint effects of model parameters on the relative efficacy of alternative interventions. In this paper we address treatment programs, i.e., the intervention strategies that result in moving infected individuals either into a different group of infected (with different biological or behavioral characteristics) or into a group of not-infected, which can correspond to susceptible, recovered or any other group that consists of not contagious individuals.

When applied to one or several groups of infected individuals, the action of a treatment strategy can be described by a parameter $u$ that corresponds to a fraction of potentially eligible individuals that are administered the treatment during a unit of time. Typically, this fraction is rather small. Our first result consists in defining the notion of a controlled reproduction number $R_0(u)$ which explicitly contains $u$ as a parameter. The next result allows to estimate the action of the particular treatment strategy on the value of the controlled reproduction number $R_0(u)$. We formulate conditions under which the application of a treatment strategy leads to the reduction in $R_0(u)$. This allows for evaluating the efficacy of a devised treatment strategy as well as for comparing the efficacy of different strategies. The obtained results are illustrated by a number of examples. 

The paper is organized as follows: Section \ref{sec:dynamics} describes a general compartmental epidemic model and briefly introduces the NGM method for computing the basic reproduction number $R_0$. The main result of the paper along with a discussion on further extensions and ramifications of the developed approach are presented in Sec.~\ref{sec:R1}. The obtained results are illustrated by a number of examples in Sec.~\ref{sec:exa}. The paper is concluded with a discussion section. Finally, there are two appendices containing necessary technical information.

\section{Epidemics dynamics}\label{sec:dynamics}
\subsection{A disease propagation model}

When modeling the process of disease propagation in a heterogeneous population, the standard procedure consists in dividing the total population into a number of classes (compartments) according to some criteria relevant to the disease transmission: disease status, ability to contract disease, behavioral (contact) pattern and so on. All individuals within a compartment are assumed to be identical in their evolution. Thus we can consider only the number of individuals within each compartment, denoted by $x_i$, where $i=1,\ldots,n$ is the number of the compartment. 

\paragraph{General formulation} 
The dynamics of these groups include transitions between groups and the in- and outflows associated with these groups. In the following it is assumed that the total inflow is constant over time and does not depend on the population size while the outflow depends on both the size and the structure of the population. The evolution of the $i$th state is thus described by the following compartmental DE:
\begin{equation}\label{eq:compart}
\dot{x}_i=\Phi_i(x)=w_i+\sum\limits_{i\neq j}\big(a_{ij}(x)-a_{ji}(x)\big) - a_i x_i,
\end{equation} 
where $x\in \R^n_{\ge 0}$ is the state, $w_i$ is the constant inflow, $a_i x_i$ is the outflow from the $i$-th compartment, and $a_{ij}(x)$, $i\neq j$, is the flow rate from the $j$-th to the $i$-th compartment.

The flow rate functions are assumed to be $C^\infty$ for all $x\in \R^n_{\ge 0}$. Furthermore, all flow rates have to satisfy the following properties:
\begin{subequations}\label{eq:prop}
\begin{align}
&w_i\ge 0, a_i\ge 0, a_{ji}(x)\ge 0, &&\forall x\in \R^n_{\ge 0}, i,j=1,\dots,n.\label{eq:prop1}\\
&x_i=0 \implies a_{ji}(x)=0 &&\forall i,j=1,\dots,n.\label{eq:prop3}
\end{align}
\end{subequations}
The property (\ref{eq:prop3}) implies that there is no outflow from an empty compartment.

\paragraph{Compartmental epidemic model}
All state variables are divided into two groups depending on whether the respective compartment corresponds to the infectious individuals (that is those able to transmit infection) or to the non-infectious individuals (regardless of whether they were infected earlier or not).

Let there be $l$ state variables representing the ``infectious'' compartments. We rearrange the states in the way that the vector of the state variables takes the form $x^\top=[x^\top_I, x_N^\top]$, where $x^\top_I=[x_1,\dots,x_l]$ and $x^\top_N=[x_{l+1},\dots,x_n]$ are the states associated with ``infectious'' and ``non-infectious'' compartments. We will refer to the respective compartments as the $I$- and $N$-compartments. Further, we write $\Phi(x)$ as the sum of two vector-valued functions
\begin{equation}\label{eq:Phi-split}\Phi(x)=\Phi^I(x) + \left[w +\Phi^{in}(x) - \Phi^{out}(x)\right]=\Phi^I(x) + \Phi^C(x),\end{equation}
where $\Phi_i^{out}(x)=\sum\limits_{j\neq i}a_{ji}(x)+a_ix_i$, and
\begin{equation*}
\Phi^I(x)=\begin{bmatrix}\bigg(\sum\limits_{j>l}a_{ij}(x)\bigg)_{i=1,\dots,l}\\[20pt]{\mathbf 0}\end{bmatrix},\enskip 
\Phi^{in}(x)=\begin{bmatrix}\bigg(\sum\limits_{\substack{j\le l\\j\neq i}}a_{ij}(x)\bigg)_{i=1,\dots,l}\quad\\[10pt] \bigg(\sum\limits_{j\neq i}a_{ij}(x)\bigg)_{i=l+1,\dots,n}\end{bmatrix}.\end{equation*} 
Here, $\Phi_i^I$ is the rate at which new infections occur in the $i$-th compartment (new infection means that the flow is to be from an $N$-compartment), $\Phi^{in}(x)$ is the vector of inter-compartmental inflow rates not related to new infections, and $\Phi^{out}(x)$ are the outflow rates from the respective compartments. Thus, $\Phi^{C}(x)$ represents the totality of all flows not related to new infections. Note that $\Phi^I_i(x)=0$ for all $i>l$ as there is no inflow of infected into the $N$-compartments.

The components of the respective vectors $\Phi^{(\cdot)}(x)$ satisfy the following properties (note that (\ref{eq:prop-Phi1}) and (\ref{eq:prop-Phi2}) can be derived from (\ref{eq:prop})):
\begin{subequations}\label{eq:prop-Phi}
\begin{align}
&\Phi_i^I(x), \Phi_i^{in}(x), \Phi_i^{out}(x)\ge 0,&&\forall i=1,\dots,n\label{eq:prop-Phi1}\\
&\{x_i=0\} \implies \{\Phi_i^{out}(x)=0\},&& \forall i=1,\dots,n\label{eq:prop-Phi2}\\
&\{x_I=\mathbf{0}\} \implies \{\Phi_i^I(x)=0, \Phi_i^{in}(x)=0\},&& \forall i=1,\dots,l. \label{eq:prop-Phi3}
\end{align}\end{subequations}
 The latter property implies that no new infections occur in a totally healthy population.\smallskip

\paragraph{Assumptions}
We conclude this section by making a number of epidemiologically motivated assumptions.
\begin{itemize}
\item[\bf A1. ] There is no steady inflow into the $I$-compartments, i.e., $w_i=0$, $i=1,\dots,l$.
\item[\bf A2. ] There exists a positive constant $\mu$ such that $a_i\ge \mu>0$ for all $i=1,\dots,n$. We refer to $\mu$ as the {\em baseline mortality rate}.
\end{itemize}

\subsection{Stability of a disease-free equilibrium}

Let $\XX_{DF}\subset \R^n_{\ge 0}$ be the set of disease-free states, i.e., $\XX_{DF}=\{x\in \R^n_{\ge 0}| x_i=0, i=1,\dots,l\}$. 

\begin{defn}
Let $x^*$ be an equilibrium state, i.e., $\Phi(x^*)=0$. Then $x^*$ is said to be a {\em disease-free equilibrium} (DFE) if $x^*\in \XX_{DF}$. Otherwise, $x^*$ is referred to as an {\em endemic infection equilibrium}.
\end{defn}

The system (\ref{eq:compart}) is locally stable at a DFE $x^*$ if the linearized model $\dot{\tilde x}=A\tilde x$ is asymptotically stable. The latter holds if the spectrum of the structure matrix $A$ contains only the eigenvalues with negative real part, \cite{Coddington:55}. Such matrices are said to be {\em Hurwitz}.

\begin{lem}[\cite{Driessche:02}]\label{lem:linearize-DFE} Let $x^*$ be a DFE. The Jacobian matrices $A^I=D\Phi^I(x)\big|_{x=x^*}$ and $A^C=D\Phi^C(x)\big|_{x=x^*}$ have the following form:
$$A^I=\begin{bmatrix}A^I_{11}&0\\[10pt]0&0\end{bmatrix},\quad 
A^C=\begin{bmatrix}A^C_{11}&0\\[10pt]A^C_{21}&A^C_{22}\end{bmatrix},$$
where $A^I_{11}$ is an $[l\times l]$ matrix with non-negative elements, $A^I_{11}\succeq 0$ and $A^C_{11}$ is an $[l\times l]$ matrix with non-positive diagonal elements and non-negative off-diagonal elements\footnote{See Appendix B for the explanation of the used notation.}. 
\end{lem}

At this point, we make one more assumption regarding the behavior of the system in the absence of the infection.

\begin{itemize}
\item[\bf A3. ] $A^C$ is {\em Hurwitz}.
\end{itemize}
This assumption effectively implies that the DFE is asymptotically stable provided the virus or whatever source of the infection has lost its contagiousness hence, no new infections occur. 

Since the block matrix $A^C_{22}$ is stable at a DFE, stability of the whole system is determined by the sum $A^I_{11}+A^C_{11}$. Following the common convention we write $F=A^I_{11}$ and $V=-A^C_{11}$. Hence, the stability of a DFE is determined by the stability of the matrix $F-V$.

\begin{thm}[\cite{Driessche:02}]\label{thm:rho}
The matrix $F-V$ is Hurwitz if and only if $\rho(FV^{-1})<1$.
\end{thm}

We conclude this section with the following two definitions:

\begin{defn}
The matrix $FV^{-1}$ is called the {\em the next generation matrix}.
\end{defn}

\begin{defn}\label{def:R0}
The parameter $R_0=\rho(FV^{-1})$ is called the {\em basic reproduction number}.
\end{defn}

In Def.\ \ref{def:R0}, $\rho(\cdot)$ denotes the spectral radius of a matrix. 
\section{Evaluating the Effect of a Treatment-type Control}\label{sec:R1}

\subsection{A controlled disease propagation model}
We wish to model the effect of treatment-type controls, i.e., the controls that result in moving individuals from the infected compartments to other (both infected and healthy) compartments. In doing so, we restrict our attention to the cases where the control action enters the equations linearly. In particular, this implies that there is no interference between different treatment strategies, i.e., the rate at which people are administered to treatment $i$ does not depend on the respective rate associated with treatment $j$ for all pairs of $i$ and $j$.

With the above assumptions in mind, we write the controlled population balance model in the following form:
\begin{equation}\label{eq:compart-u}\dot{x}=\Phi(x)+\Phi^u(x) u,\end{equation}
where $u\in \R^m_{\ge 0}$ is the vector of non-negative controls and the components of the control matrix $\Phi^u(x)=\left[\Phi_{ik}^u(x)\right]$, $i=1,\dots,n$, $k=1,\dots, m$, have the following structure:
$$\Phi^{u}_{\cdot,k}(x)=\begin{bmatrix}
\left(\sum\limits_{\substack{j\le l\\j\neq i}}a^u_{ij,k}(x)-\sum\limits_{j\neq i}a^u_{ji,k}(x)\right)_{i=1,\dots,l}\\
\left(\sum\limits_{j\le l}a^u_{ij,k}(x)\right)_{i=l+1,\dots,n}\end{bmatrix}.$$
The control matrix describes the controlled flows from the infected compartments to both infected and healthy compartments. This may correspond to isolating infected individuals (corresponds to a transition from an $I$-compartment  to another $I$-compartment) or treating an infected individual (a flow from an $I$-compartment to an $N$-compartment). 

The respective controlled flow rates $a^u_{ij,k}(x)$ satisfy the conditions similar to those for the flow rates of the original (uncontrolled) model:
\begin{subequations}\label{eq:prop-u}
\begin{align}
&a^u_{ji,k}(x)\ge 0, &&\forall x\in \R^n_{\ge 0}, i,j=1,\dots,n,\,i\neq j,\, k=1,\dots,m.\label{eq:prop-u1}\\
&\{x_i=0\} \implies \{a^u_{ji,k}(x)=0\} &&\forall  i,j=1,\dots,n,\,k=1,\dots,m.\label{eq:prop-u2}
\end{align}
\end{subequations}

\paragraph{Controllability}
Note that at a DFE $x^*$, the control matrix vanishes identically as there are no outflows from I-compartments: $\Phi^{u}(x^*)=0_{[n\times m]}$. This implies that at a DFE $x^*$, the linearized model is effectively uncontrollable. Thus, one cannot use the linearized model to design a feedback control law stabilizing the DFE.

\subsection{Analysis of the controlled system}\label{sec:const-control}

Due to the difficulties outlined above, we consider a somewhat simpler, but even more practically relevant situation. Namely, we wish to study the system's dynamics when a constant control $u^*\in \R^m_{\ge 0}$ is applied. We first formulate the result that extends Lemma \ref{lem:linearize-DFE} to include the controls $u$.
\begin{lem}\label{lem:linearize-u}
Let $u(t)=u^*\in \R^m_{\ge 0}$ for all $t\ge 0$. The linearized model of (\ref{eq:compart-u}) at a DFE $x^*$ has the form
\begin{equation}\label{eq:linear-sys}\dot{x}=\left(A^I+A^C+\sum_{k=1}^m B_k u^*_k\right)x,\end{equation}
where $A^I$ and $A^C$ are defined in Lemma \ref{lem:linearize-DFE} and $B_k$, $k=1,\dots,m$ are
$$B_k=\frac{D\Phi_{\cdot,k}^{u}(x)}{Dx}\bigg|_{x=x^*}= \begin{bmatrix} B_{11,k}&0\\[5pt]B_{21,k}&0\end{bmatrix}.$$
The matrices $B_{11,k}$ have non-positive diagonal and non-negative off-diagonal elements. Furthermore, $B_{11,k}$ are weakly column diagonally dominant.
\end{lem}
\begin{proof}
See Appendix A.
\end{proof}

For the sake of notational simplicity and to comply with the previously accepted notation, we will write $W_i=-B_{11,i}$. Matrices $W_i$ belong to the class of Z-matrices (see Appendix B for details). However, these are not M-matrices as  matrices $W_i$ are typically rank-deficient and hence non-invertible.

We are interested in the stability of the linearized system (\ref{eq:linear-sys}). As in the uncontrolled case, the structure matrix of (\ref{eq:linear-sys}) is a block-lower triangular matrix, whose eigenvalues coincide with the eigenvalues of the diagonal blocks. The lower right block is Hurwitz according to {\bf A3}. Thus, the stability of (\ref{eq:linear-sys}) is determined by the eigenvalues of the matrix $$J_u=F-V-\sum_{k=1}^m W_k u_k.$$

The following lemma gives an important result that will be used in the sequel. 
\begin{lem}
Let $V$ be a strictly column diagonally dominant Z-matrix with positive diagonal elements. Then for any $u\in \R^m_{\ge 0}$ the matrix $V+\sum_{k=1}^m W_k u_k$ is an M-matrix.
\end{lem}
\begin{proof}
According to Lemma \ref{lem:linearize-u}, matrices $W_i=-B_{11,i}$ have positive diagonal elements and are weakly (non-strictly) column diagonally dominant.  Thus for any positive $u$ the matrix $V+\sum_{k=1}^m W_k u_k$ is a Z-matrix and is strictly column diagonally dominant. Then it follows from Theorem \ref{thm:M-charact} that $V+\sum_{k=1}^m W_k u_k$ is a non-singular M-matrix.
\end{proof}

Now we are ready to formulate a generalized version of Theorem \ref{thm:rho}
\begin{thm}\label{thm:rho-u}
The matrix $F-V-\sum_{k=1}^m W_k u_k$ is Hurwitz if and only if $$\rho\left(F\left(V+\sum_{k=1}^m W_k u_k\right)^{-1}\right)<1.$$
\end{thm}
We thus define the controlled reproduction number $R_0(u)$ as the spectral radius of the perturbed matrix $Q(u)=F\left(V+\sum_{k=1}^m W_k u_k\right)^{-1}$:
\begin{equation}\label{eq:R_0u}R_0(u)=\rho\left(F\left(V+\sum_{k=1}^m W_k u_k\right)^{-1}\right).\end{equation}
Obviously, we have $R_0(0)=R_0$.

Let the uncontrollable system be such that $R_0>1$. Theorem~\ref{thm:rho-u} allows for determining if a given constant control $u^*$ suffices to shift the value of the basic reproduction number in order to make it less than 1. However, in many cases it is difficult to compute the perturbed reproduction number $R_0(u)$. Thus, we would like to have a result that would tell us if a given structure of the treatment allows for achieving the stated goal.

The following result provides the required estimation. We first consider the case $m=1$, i.e., we assume that there is a scalar control $u> 0$.
\begin{thm}\label{thm:R1}
Let $R_0$, $x_0$ and $y_0$ be the spectral radius of $FV^{-1}$ as well as the right and the left eigenvectors of $FV^{-1}$ corresponding to $R_0$, respectively. Let, furthermore, there be only one eigenvalue of $FV^{-1}$ coinciding with $R_0$ and other eigenvalues be strictly less than $R_0$ in absolute value. For sufficiently small $u$, the sign of variation $R_0(u)-R_0$ is determined by the sign of 
\begin{equation}\label{eq:R1}R_1= - y_{0}^\top V^{-1}W x_0 R_0/(y_{0}^\top x_0)\end{equation}
\end{thm}
\begin{proof}
The eigenvalues of $Q(u)=F\left(V+W u\right)^{-1}$ change continuously with $u$. Hence, we can write the spectral radius $R_0(u)=R_0+uR_1 +O(u^2)$ and the respective eigenvector as $x(u)=x_0+u x_1 +O(u^2)$. We thus have
\begin{equation}\label{eq:eig-pert}F(V+u W)^{-1} (x_0+u x_1 + O(u^2))=(R_0+u R_1 + O(u^2)) (x_0+u x_1 + O(u^2)).\end{equation}

First, we note that $\frac{d}{du}(V+u W)^{-1}=-(V+u W)^{-1}W(V+u W)^{-1}$ and at $u=0$ we have $\frac{d}{du}(V+u W)^{-1}\bigg|_{u=0}=-V^{-1}WV^{-1}$. Differentiating the left and the right sides of (\ref{eq:eig-pert}) and setting $u=0$ we get:
\begin{equation}\label{eq:eig-pert2}(FV^{-1} - R_0 I)x_1=(FV^{-1}WV^{-1} + R_1 I) x_0.\end{equation}
The matrix $(FV^{-1} - R_0 I)$ has a zero eigenvalue and the respective left eigenvector is $y_0$, i.e., $y_0^\top(FV^{-1} - R_0 I)=0$. Multiplying both sides of (\ref{eq:eig-pert2}) with $y_0^\top$ and expressing $R_1$ we get:
$$R_1 = - y_{0}^\top WV^{-1} x_0 R_0/(y_{0}^\top x_0).$$
Finally, we note that for sufficiently small (and positive) $u$ the sign of the difference $R_0(u)-R_0$ is determined by the sign of $R_1$, whence the result follows.
\end{proof}
That is to say, we can determine if a given treatment is efficient (at least for small values of $u$) by checking the sign of $R_1$. If this sign is negative then by increasing $u$ we decrease $R_0(u)$ and eventually ensure that it becomes less than 1. Otherwise (if $R_1>0$) we conclude that the treatment program is inadequately formulated and does not lead to a decrease in the reproduction number. See Sec.\ \ref{sec:exa} for examples.

Following the same procedure, one can write $R_0(u)$ for the case when $u\in \R^m$. We have 
$$R_0(u^*)=R_0+\sum_{k=1}^m u^*_k R^k_1 +O(\|u^*\|^2).$$
If the components of $u^*$ are sufficiently small, the contribution of each individual component is determined by the respective term $R^k_1$, which is defined as 
\begin{equation}\label{eq:R_1i}R^k_1 = - y_{0}^\top W_k V^{-1} x_0 R_0/(y_{0}^\top x_0).\end{equation}

We immediately arrive at the following result.
\begin{thm}
For sufficiently small values of the controls $u_i$, the effect of each control is independent from the values of the remaining controls. The total change of $R_0(u)$ is equal to the sum of individual contributions up to the high order term: $R_0(u)-R_0=\sum_{k=1}^m u_k R^k_1 +O(\|u\|^2)$. 
\end{thm}

This result indicates the second potential use of $R_1$. Let there be $m$ treatment strategies that aim at decreasing the value of $R_0(u)$. Then the most efficient strategy is the one with the smallest (and necessarily negative) value of the parameter $R^k_1$. 

\paragraph{Discussion} Note that the structure of the expression for $R_1$ does resemble that one for $R_0$ and it can be interpreted as follows. First, following \cite{Driessche:02}, we interpret the $(i,j)$ entry of $V^{-1}$ as the average duration of time an individual introduced into the $j$-th compartment spends in the $i$-th one, assuming there is no reinfection and no control. The $(i,i)$ entry of $(-W_k)$ is the rate at which infected individuals are removed from the $i$-th compartment, while the $(i,j)$-th component, $i\neq j$, corresponds to the rate at which the infected from $i$-th compartment are moved to the $j$-th compartment. The respective rates are multiplied by the control $u_k$. 

Hence, the $(i,j)$-th element of $-WV^{-1}$ is the relative treatment-induced rate of flow (outflow if negative or inflow if positive) from the $i$-th compartment as applied to the individuals initially introduced into the $j$-th compartment. Note that $\sum_{i}(-WV^{-1})_{i,j}=0$ if the treatment consists in redistributing the infected individuals between infected compartments and is negative if at least a part of infected individuals are removed to the susceptible compartments as a result of treatment.

To interpret the right eigenvector of $FV^{-1}$, corresponding to $R_0$, recall that $FV^{-1}x_0$ is the expected number of new infections produced by the initial distribution of infected individuals given by $x_0$. Hence, the right eigenvector $x_0$ can be interpreted as the worst case distribution of the infected. In contrast to that, the left eigenvector $y_0$ can be seen as a worst case transmissibility rates, i.e., the transmissibility rates that result in the maximal infection spread in the population taking into account the existing structure of the transmission routes. Note that the eigenvectors are defined up to a positive factor and hence should be seen as proportions rather than absolute values.

Finally, the expression $- y_{0}^\top WV^{-1} x_0 (y_{0}^\top x_0)^{-1}$ can be interpreted as the total flow of the infection due to the  treatment computed for the worst case scenario. The term  $(y_{0}^\top x_0)^{-1}$ serves as a normalizing factor. If the total flow is negative, the treatment leads to a decrease in infection and to an increase otherwise. Note that $R_1$ is negative if $\sum_{i}(-WV^{-1})_{i,j}<0$ for all $j$.   

\begin{rem}There are two main advantages of considering $R_1$, resp., $R_1^k$ in contrast to working directly with $R_0(u)$ as defined in (\ref{eq:R_0u}):
\begin{enumerate}
\item The function $R_0(u)$ depends on $u$ in a highly nonlinear way, making it difficult if ever possible to analyze the impact of the control on the controlled reproduction number $R_0(u)$. In contrast, $R_1$ is  a well-defined quantity that does not depend on $u$ thus making it more amenable for analysis.
\item In general, the matrix $\left(V+\sum_{k=1}^m W_k u_k\right)^{-1}F$ has a complex structure which substantially complicates the problem of finding its spectral radius in an analytic form. In contrast to that, $R^k_1$ can be computed with much less effort. Also, we note that the vectors $x_0$ and $y_0$ need to be computed only once. The modeler can then vary the structure of $\Phi^u$ (and, respectively, $W_k$) to achieve the required effect in $R^k_1$.
\end{enumerate} 
\end{rem}

\begin{rem}
Note that the preceding results are formulated for sufficiently small values of $u$. However, this assumption is not as restrictive as it may seem. The point is that each component of the control $u$ denotes a fraction of the set of potentially eligible individuals that is enrolled into the respective treatment during a unit of time. In most practically relevant cases, this fraction is rather small, ranging from thousandths to hundredths, thus justifying the assumption. 
\end{rem}

\subsection{Further extensions}

The described method is substantially based upon the eigenvalue perturbation theory. Namely, we analyze the behavior of the largest real eigenvalue (spectral radius) of a next generation matrix under certain structured perturbations. All the presented results are of local nature, i.e., these are valid for sufficiently small values of $u$. However, in many cases it is possible to extend this analysis to larger values of $u$. Below, we consider several possible scenarios and consider possible options and potential limitations.

The simplest case is when the rank of the matrix $F$ is equal to 1. This corresponds to the situation when all immediately infected individuals enter the same compartment, but the infection itself can be caused by a contact with different types of infected. In this case, there is only one non-zero positive eigenvalue coinciding with the spectral radius. The only possible limitation corresponding to this situation is that the controlled reproduction number can change its behavior when the control $u$ grows sufficiently large.  One could thus wish to study the behavior of $R_0(u)$ in some more detail. To do so one can compute the second term in the series expansion which we denote by $R_2$. The corresponding expression is presented in Appendix A, Eq. (\ref{eq:r2}). However, determination of $R_2$ is a rather cumbersome procedure involving computing Moore-Penrose pseudo inverse. On the other hand, it seems that most cases the sign of $R_1$ describes the behavior of $R_0(u)$ for arbitrary large values of $u$ and hence can be used for making global predictions.

When the rank of $F$ is larger than 1, the matrix $(V+uW)^{-1}F$ has in general more than one nonzero eigenvalue. A model of co-infection and two types of treatment, one for each type of infection, is a typical example of such a situation. In this case, there are two nonzero eigenvalues $\lambda_1(u_1)$ and $\lambda_2(u_2)$ controlled by the respective treatments $u_1$ and $u_2$. Suppose that for $u_1=u_2=0$, $R_0(0)=\lambda_1(0)>\lambda_2(0)$. Hence, the optimal strategy (in terms of minimizing $R_0(u)$) would be to invest into $u_1$. However, as $\lambda_1(u_1)$ decreases, it will at some point equate with $\lambda_2(u_2)$. Further investment into $u_1$ will lead to the change of the roles: $\lambda_2(u_2)$ will become greater than $\lambda_1(u_1)$ and hence will account for $R_0(u)$.  

In general, the behavior of the eigenvalues can be predicted by analyzing the structure of $Q(u)=(V+uW)^{-1}F$. We assume here that this structure does not change with $u$, that is, $Q_{ij}(u) = 0, u> 0 \Rightarrow Q_{ij}(u) = 0, \forall u>0$ and $Q_{ij}(u)\neq 0, u> 0 \Rightarrow Q_{ij}(u)\neq 0, \forall u> 0$. In this case there are two possible situations: the matrix $Q(u)$ can be either irreducible or reducible. The former implies that there is a simple real eigenvalue corresponding to the spectral radius and hence to $R_0$ (see, e.g., \cite[Sec. 8.3]{Meyer:00}). If the matrix $Q(u)$ is reducible, one can find a permutation matrix $P$ such that $\bar{Q}(u)=PQ(u)P^{-1}$ has a block-diagonal form and the blocks on the main diagonal are irreducible. Since the eigenvalues of the permuted matrix are determined by the diagonal blocks $\bar{Q}_i(u)$ it suffices to track the spectral radii of $\bar{Q}_i(u)$; denote them by $\bar{\lambda}_i(u)$. We have that $R_0(u)=\max_i \left(\bar{\lambda}_i(u)\right)$. If for some $u^*$ it happens that $\bar{\lambda}_i(u^*)=\bar{\lambda}_j(u^*)$ we have the situation described in the preceding paragraph. 
\section{Examples}\label{sec:exa}

In this section we consider three sufficiently simple but yet non-trivial epidemiological models that are aimed at illustrating the usefulness of the proposed approach.

\subsection{An SI model with acute and chronic stages and a single treatment}\label{sec:exa1} 

\paragraph{Model}
Consider an SI-model of a disease with two stages (acute and chronic) and a treatment $u$. This general model structure with sequential infection compartments containing infected persons with different levels of contagiousness describes the evolution of such diseases as HIV \cite{WHO:05staging} and syphilis \cite{french_syphilis_2007} that have variable levels of contagiousness over the course of infection. 
For the sake of illustration we explicitly show the $\Phi^I$, $\Phi^C$, and $\Phi^u$ components in (\ref{eq:SI}).
\begin{equation}\label{eq:SI}
\frac{d}{dt}\begin{bmatrix}I_A\\I_C\\T\\S\end{bmatrix}= \underbrace{\begin{bmatrix}\alpha(X)S\\0\\0\\0\end{bmatrix}}_{\Phi^I(x)}+
\underbrace{\begin{bmatrix}-\beta I_A-\mu_A I_A\\\beta I_A+\delta  T - \mu_C I_C\\-\delta  T - \mu_T T\\w-\alpha(X)S-\mu_S S\end{bmatrix}}_{\Phi^C(x)}+ \underbrace{\begin{bmatrix}0\\-I_C\\I_C\\0\end{bmatrix}}_{\Phi^{u}(x)}u
\end{equation}
where $\alpha(X)=\dfrac{\alpha_A I_A + \alpha_C I_C + \alpha_T T}{N}$, $X=[I_A, I_C, T, S]$, and  $N=I_A+I_C+T+S$. 
The states correspond to the number of acutely infected ($I_A$), chronically infected ($I_C$), treated ($T$), and susceptible ($S$) individuals. We assume that the first three compartments are infectious with different transmission probabilities: $\alpha_A$, $\alpha_C$, and $\alpha_T$. Furthermore, $\beta$ is the inverse duration of the acute phase, $\mu_S$, $\mu_A$, $\mu_C$ and $\mu_T$ are the mortality rates for susceptible, non-treated in acute and chronic phases and treated infected, with $\mu_C>\{\mu_A,\mu_T\}>\mu_S$, and $\delta$ is the rate at which the treatment fails. 

\paragraph{Analysis}
The disease-free equilibrium is unique and given by $x^*=[0, 0,0, w/\mu_S]$.
%
The (uncontrolled) value of $R_0$ is computed as the spectral radius of $FV^{-1}$ and is equal to
$$R_0=\frac{\alpha_C \beta + \alpha_A \mu_C}{\left(\beta + \mu_A\right)\mu_C}
$$
Note that the DFE $x^*$ is unstable if $\alpha_C \beta + \alpha_A \mu_C > (\beta  + \mu_S)\mu_C$. We wish to evaluate if the proposed treatment strategy is efficient. To do so we compute $R_1$ according to (\ref{eq:R1}):

\begin{equation}\label{eq:R0u-ex1}R_1=-\frac{\beta \left(\alpha_C \mu_T - \alpha_T \mu_C\right)}{\mu_C^2 \left(\beta + \mu_A\right) \left(\delta + \mu_T\right)}
\end{equation}
%
%
This expression is negative if $\alpha_C \mu > \alpha_T \mu_C$. 
This can be written as 
\begin{equation}\label{eq:cond-ex1}\alpha_C \mu_C^{-1} > \alpha_T \mu_T^{-1},\end{equation}
that is the treatment turns out to be efficient if the product (transmission probability times average residence time in the respective compartment) is lower for the treated individuals as compared with chronically infected. In the following, we will refer to this product as the {\em cumulative transmissibility} of the infection induced by a given compartment. The condition (\ref{eq:cond-ex1}) can be interpreted as follows: If treated people live longer than untreated (but infected), but their infectivity does not reduce enough to compensate for this, such treatment will contribute to the propagation of the disease.  

 
For the considered model it is possible to compute the controlled reproduction number $R_0(u)$:
$$R_0(u)=\frac{\left(\alpha_C \beta + \alpha_A \mu_C\right) \left(\delta + \mu_T\right) + \left(\alpha_T \beta + \alpha_A \mu_T\right) u}{(\beta + \mu_A) (\delta \mu_C + \mu_C \mu_T + \mu_T u)}.
$$   
One may check that $R_1=\dfrac{\partial R_0(u)}{\partial u}\bigg|_{u=0}$. However, we are interested in the behavior of $R_0(u)$ for (relatively) large values of $u$. To do so, we compute the difference 
\begin{equation}\label{eq:r0-diff-SII}
R_0(u)-R_0=-\frac{\beta u \left(\alpha_C \mu_T - \alpha_T \mu_C\right)}{\mu_C \left(\beta + \mu_A\right) \left(\delta \mu_C + \mu_C \mu_T + \mu_T u\right)}.
\end{equation}
One can readily observe that the difference is negative for any $u>0$ when (\ref{eq:cond-ex1}) holds.

Figure \ref{fig:ex2} shows the relative error in the $R_1$ approximation to $R_0(u)$ for a set of random parameters as a function of $R_0$.
We obtained random parameter sets for a given value of $R_0$ optimizing from a random starting point for $\alpha_A$, $\alpha_C$, and $\beta$ assuming $\mu_A = \mu_C = 120^{-1}$, $\mu_T = 360^{-1}$, and $\alpha_T=0$.
We rejected parameter sets such that 
there was no feasible control capable of reducing $R_0(u)$ to 1.
The relative error is defined as the difference in the exact value of $u^*$ such that $R_0(u^*)=1$ and $u^{**}=\frac{R_0-1}{R_1}$ normalized by $u^*$. 
The approximate value of the control is always less than the true value and the approximation becomes worse as the value of $R_0$ increases, although the accuracy of the approximation is dependent on the specific parameter values. 
However, if $R_0$ is small the approximation gives a reasonable indicator of the level of control required to bring the epidemic to the threshold level.

\begin{figure}[h]
\centering
\includegraphics[width=10cm]{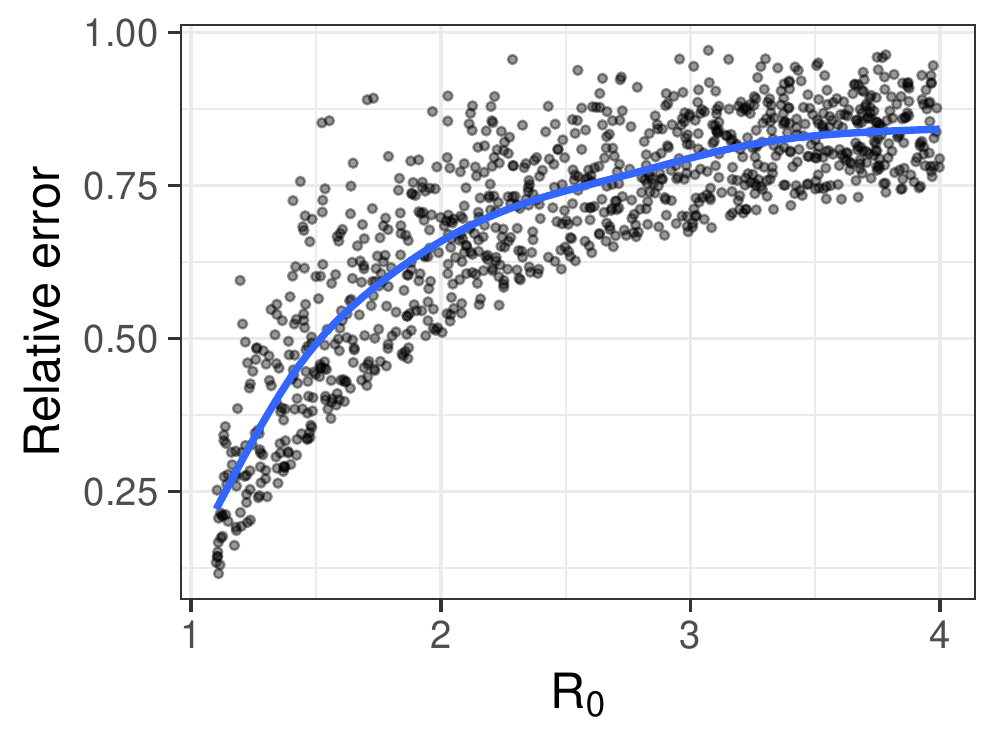}
\caption{Relative error in approximate reduction in $R_0$ due to treatment. This figure shows the multiplicative-scale error defined as $\frac{u^*-u^{**}}{u^*}$ where $u^{*}$ such that $R_0(u^*)=1$ and $u^{**}=\frac{1-R_0}{R_1}$ for a set of randomly selected parameter values. Parameters with a given $R_0$ between 1.1 and 4 were selected at random as described in the text. The blue line is a local polynomial regression showing the general trend. 
}
\label{fig:ex2}
\end{figure}

\subsection{An SEIR model with an asymptomatic stage and treatment}\label{sec:exa2}

\paragraph{Model}
Consider the model shown in Fig.\ \ref{fig:SEIR}. This model describes the transmission dynamics of an asymptomatic (sub-clinical) infection. Individuals with asymptomatic
infection do not develop the respective symptoms, but are infectious and contribute to the distribution of the disease.
 Asymptomatic infection has been shown to exist for many diseases, including, e.g., herpes \cite{Wald:95}, gonorrhea \cite{Korenromp:02}, measles \cite{Anlar:02}, and common cold.

\begin{figure}[tbh]
  \centering
    \includegraphics[width=0.75\textwidth]{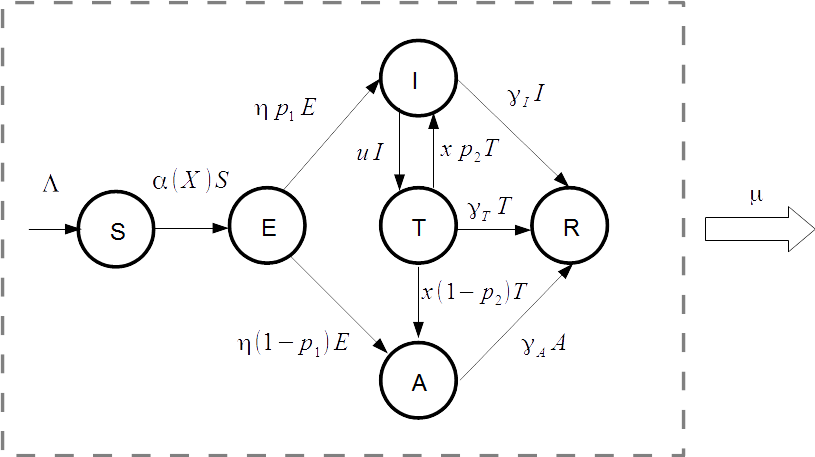}
     \caption{SEIR model.
}
      \label{fig:SEIR}
\end{figure}

Here, we have the following compartments: $E$ -- exposed, $I$ -- infected, $A$ -- asymptomatically infected, $T$ -- treated. $\Lambda$ is the inflow into the population, $\eta$ is the transmission rate, $p_1$ is the probability that the exposition will lead to the typical clinical course and $p_2$ is the probability that the failed treatment will result in the typical clinical course. Furthermore, $x$ is the rate at which the treatment is failed or canceled, and $\gamma_I$, $\gamma_A$ and $\gamma_T$ are the inverses of the mean residence time in respective compartments. We assume that only the symptomatic infected, i.e., $I$, are treated. However, the disease can be transmitted both by symptomatic and asymptomatic infected though with different transmission rates $\alpha_I$ and $\alpha_A$. People under treatment are assumed to be non-infectious.

Its behavior is described by the following DEs:
$$\begin{cases}
\dot{E}=\alpha(X) S - (\eta + \mu)E \\
\dot{A}= - (\gamma_A + \mu) A + \eta (1 - p_1) E + x (1 - p_2) T\\
\dot{I}= \eta p_1 E - \left(\gamma_I + \mu\right) I - u I + p_2 x T\\
\dot{T}= u I - (\gamma_T + \mu  + x) T\\
\dot{S} =\Lambda - \mu S - \alpha(X) S\\
\dot{R}= \gamma_A A + \gamma_I I + \gamma_T T - \mu R, \end{cases}
$$
where $\alpha(X)=\dfrac{\alpha_A A + \alpha_I I}{N}$.

\paragraph{Analysis}
The basic reproduction number $R_0$ is found using the NGM approach to be
$$R_0=\frac{\eta}{\eta + \mu}\left[
\frac{\alpha_A}{\gamma_A + \mu}(1 - p_1)+\frac{\alpha_I}{\gamma_I + \mu}p_1\right],
$$
where the expression in square brackets is the expected value of the total cumulative transmissibility of the infection and the quantity in front of the brackets is the fraction of the infected that leave $E$ toward either $I$ or $A$.

We compute $R_1$
\begin{equation}\label{eq:r1-seiar}R_1=-\dfrac{\eta p_1 \big[\alpha_I (\gamma_A + \mu)(\gamma_T + \mu) + 
\left(\alpha_I (\gamma_A + \mu) - \alpha_A (\gamma_I  + \mu)\right)(1 - p_2) x\big]}{\left(\eta + \mu\right) \left(\gamma_{A} + \mu\right) {\left(\gamma_I + \mu\right)}^2 \left(\gamma_{T} + \mu + x\right)}\end{equation}
which turns out to be negative if the numerator is positive (note that the denominator is always positive). 


Upon some algebraic manipulations, this condition can be written as
\begin{equation}\label{eq:r1-seiar-cond}\frac{\alpha_I}{\gamma_I + \mu}> \frac{x (1 - p_2)}{\gamma_T + \mu} \left[\frac{\alpha_A}{\gamma_A + \mu} - \frac{\alpha_I}{\gamma_I + \mu}\right].\end{equation}
Define the cumulative transmissibilities of the $A$ and $I$ stages as $\zeta_A=\frac{\alpha_A}{\gamma_A + \mu}$ and $\zeta_I=\frac{\alpha_I}{\gamma_I + \mu}$. With this, the condition (\ref{eq:r1-seiar-cond}) can be rewritten as 
$$\frac{\zeta_A - \zeta_I}{\zeta_A} \cdot x(1 - p_2) < \gamma_T + \mu,$$
provided $\zeta_A > \zeta_I$ (otherwise the condition (\ref{eq:r1-seiar-cond}) holds trivially as the r.h.s.\ of (\ref{eq:r1-seiar-cond}) turns to be negative).  The first term on the left side is the relative difference in cumulative transmissibility and the second one is the flow from $T$ to $A$. Finally, the expression on the right side describes the flow from $T$ to $R$ or to outside the system. This result underscores the importance of the treatment efficiency for the successful eradication of the disease.



We proceed by considering the total change of $R_0$ due to the control $u$
\begin{multline}\label{eq:r0-diff}R_0(u)-R_0=\\-\dfrac{\eta p_1 u \left[\alpha_I (\gamma_A + \mu) (\gamma_T + \mu) + \left(\alpha_I (\gamma_A + \mu) - \alpha_A (\gamma_I + \mu)\right) (1 - p_2) x\right]}{(\eta + \mu) (\gamma_A + \mu) (\gamma_I + \mu) \left[(\gamma_I + \mu)(\gamma_T + \mu + x) +  (\gamma_T + \mu + x (1 - p_2)) u\right]}.
\end{multline}
The denominator of (\ref{eq:r0-diff}) is always positive so, the condition for the treatment to be efficient (that is the condition for the difference to be negative) is that the numerator is positive. But the numerator of (\ref{eq:r0-diff}) coincides with the numerator of $R_1$, (\ref{eq:r1-seiar}), up to the control $u$ which is positive. Hence we conclude that satisfaction of the condition (\ref{eq:r1-seiar-cond}) guarantees decrease in $R_0$ for any positive value of $u$, i.e., gives a global condition.

\subsection{An SI model with high- and low-risk groups and two treatments}\label{sec:exa3}

\paragraph{Model}
Consider a simplified model of HIV transmission dynamics with two controls corresponding to treating infected individuals from the high- and the low-risk groups, denoted by $u_T$ and $u_L$. For the details on the derivation of the model see \cite{HK:15}. The transmission dynamics is described by the following set of ODEs

\begin{equation}\label{eq:SI_HLm}\begin{cases}
\dot{I}_H= - I_H \mu - I_H f_L \omega + I_L f_H \omega + \phi_H(X) S_H - I_H u_H - \mu_I I_H\\ 
\dot{I}_L= - I_L \mu + I_H f_L \omega - I_L f_H \omega + \phi_L(X) S_L - I_L u_L - \mu_I I_L\\ 
\dot{T}_H=T_L f_H \omega - T_H \mu - T_H f_L \omega  +  I_H u_H\\ 
\dot{T}_L=T_H f_L \omega - T_L \mu  - T_L f_H \omega +   I_L u_L\\
\dot{S}_H=\alpha_H - S_H \mu - S_H f_L \omega + S_L f_H \omega - \phi_H(X) S_H\\ 
\dot{S}_L=\alpha_L - S_L \mu + S_H f_L \omega - S_L f_H \omega - \phi_L(X) S_L, \end{cases}
\end{equation}
where $S_{\{H,L\}}$, $I_{\{H,L\}}$, and $T_{\{H,L\}}$ are the susceptible, infected, and treated. The subscript denotes the behavioral pattern of the respective group: there are a (H)igh and a (L)ow risk group. Further, $\phi_H(X)=\beta\lambda_H \dfrac{\lambda_H I_H + \lambda_L I_L}{\lambda_H N_H + \lambda_L N_L}$ and  $\phi_L(X)=\beta\lambda_L \dfrac{\lambda_H I_H + \lambda_L I_L}{\lambda_H N_H + \lambda_L N_L}$ are the per-capita transmission rates. $\mu$ is the mortality rate and $\mu_I$ is the disease-induced mortality; $\rho_H=f_H\omega$ and $\rho_L=(1-f_H)\omega=f_L\omega$ are the transition rates between high- and low-risk groups with $\omega$ denoting the volatility coefficient; $\lambda_H$ and $\lambda_L$ are the contact rates; $\beta$ is the infection transmissibility. Finally, the inflow rates are $\alpha_H=f_H\mu N$ and $\alpha_L=f_L\mu N$.

For the considered model, the disease free equilibrium is $[I_H^*,\,I_L^*,\,T_H^*,\,T_L^*,\,S_H^*,\,S_L^*]=[0,\,0,\,0,\,0,\,f_H N,\,f_L N]$. Note that $f_H$ and $f_L$ are the fractions of the respective (high- or low-risk) population at the DFE. Using the next generation matrix method we can compute $R_0$:

\paragraph{Analysis}
The basic reproduction number is computed using the NGM method:
$$R_0=\frac{\beta}{\left(\mu_I + \mu + \omega\right)}C_V(\lambda) + \frac{\beta}{\left(\mu_I + \mu\right)}[\lambda]
$$
where $[\lambda]=f_L\lambda_L + f_H \lambda_H$ and $[\lambda^2]=\lambda_L^2 f_L + \lambda_H^2 f_H$ are the first and the second moments of the contact rate at the DFE, $\var(\lambda)$ is the variance of the contact rate, and $C_V$ is the coefficient of variance defined as $C_V(\lambda)=\var(\lambda)\bigm/ [\lambda]$.
Since there are two controls, we would like to compare their contributions in order to decide which one should be invested into. To compute the corresponding components $R_1^H$ and $R_1^L$ we use (\ref{eq:R_1i}) to get

$$R_1^H=-\frac{\beta f_H {\left(\lambda_H(\mu_I + \mu) + [\lambda] \omega \right)}^2}{\left(\mu_I + \mu\right)^2 [\lambda] {\left(\mu_I + \mu + \omega\right)}^2}
$$

$$R_1^L=-\frac{\beta f_L \left(\lambda_L (\mu_I + \mu) + [\lambda] \omega\right)^2}{{\left(\mu_I + \mu\right)}^2 [\lambda] {\left(\mu_I + \mu + \omega\right)}^2}
$$
The first observation is that both $R_1^H$ and $R_1^L$ are negative thus, they contribute to reducing $R_0(u)$ for any choice of parameters. After some algebraic manipulations we find that $u_H$ is more efficient than $u_L$ if

$$f_H {\left(\frac{\lambda_H}{[\lambda]} + \frac{\omega}{\mu}\right)}^2 > f_L\left(\frac{\lambda_L}{[\lambda]} + \frac{\omega}{\mu}\right)^2$$
Figure \ref{fig:ex1} shows the value of $\lambda_H$ such that the two terms in the inequality above are equal. To account for the variability in the duration of the high and low-risk periods as a function of $f_H$, we introduce the normalized volatility coefficient 

$$\omega^*=\frac{\omega\mu^{-1}}{{f_H}^{-1}+(1-f_H)^{-1}},$$
which is the number of full high and low-risk episodes that can be contained in a typical infectious period.
This plot shows that for fixed value of $f_H$, behavioral volatility makes high-risk intervention more plausible (i.e.~the high-risk population does not have to be extremely high-risk to make a targeted intervention efficient). This simple analysis gives us a clear theoretical prescription for when to focus on high-risk group based on measurable aspects of the transmission system. 

\begin{figure}[h]
\centering
\includegraphics[width=10cm]{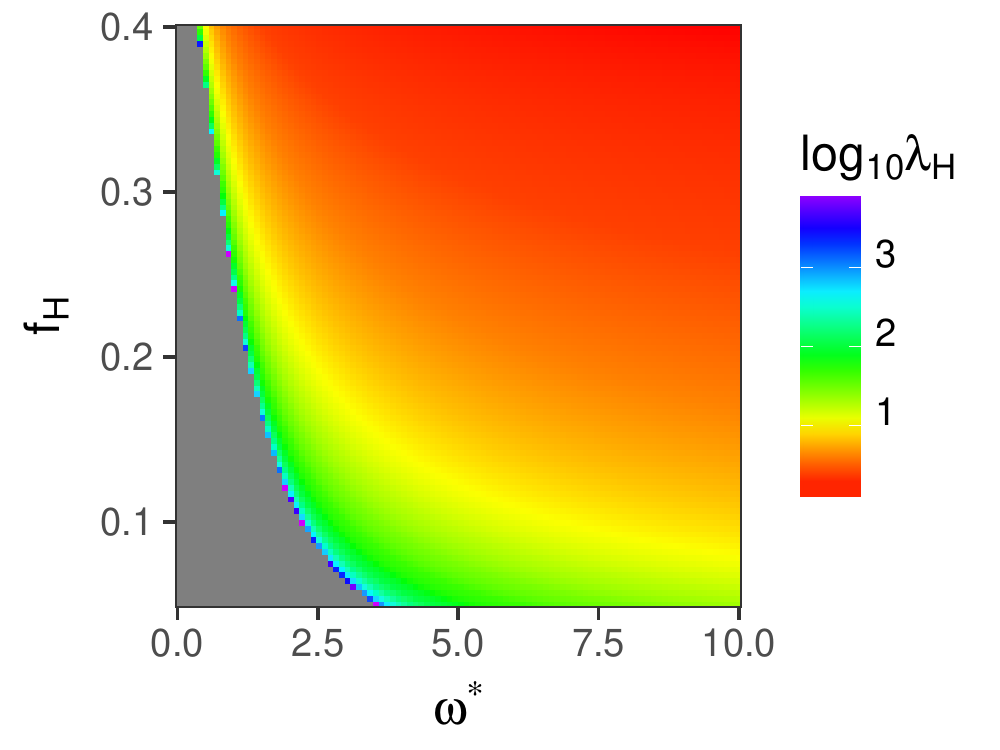}
\caption{Minimum value of the high-risk contact rate, $\lambda_H$, for which the high-risk intervention is preferred. The color shows the value of $\lambda_H$ for which the terms in the above inequality are equal. For any value of $\lambda_H$ higher than the plotted value the high-risk intervention is preferred. The gray color indicates that the value is either above $10^4$ or there is no value such that the terms are equal. The x-axis is the normalized volatility coefficient $\omega^*$. The y-axis shows the fraction of the population that is high-risk in the absence of disease. The remaining parameters are $\lambda_L=1$ and $\mu_I=0$.
}
\label{fig:ex1}
\end{figure}

\section{Discussion}


The results presented in this paper can be applied to a wide class of epidemiological models as long as their dynamics can be described by a compartmental system of form \eqref{eq:compart-u}. Our approach builds upon and further develops the next generation matrix method in that it allows one to estimate the influence of the treatment-type control(s) on the basic reproduction number $R_0$ which defines the ultimate condition for eventual elimination of a disease. The more a given control reduces $R_0$, the closer the system is to elimination and the more effective future interventions will be. Furthermore, it may turn out that in complex models an intervention could unintentionally make things worse for certain populations. The basic premise of medicine to do no harm applies to public health as well. However, the complex, non-linear dynamics of transmission limit the ability of our intuitions to predict the effects of an intervention. Likewise, measurement of the effects of interventions are often very noisy and can have long time lags. Both weak measurability of outcomes and hard to predict dynamics highlight the need for stronger theoretical guarantees that an intervention will not cause population-level harm. A possible extension to this work could include consideration of complex models of how risk behavior changes in response to changing prevalence and incidence of disease.

All our results are of local nature, i.e., these are valid for sufficiently small values of $u$. However, in many cases it is possible to extend this analysis to larger values of $u$ as was illustrated by examples  
It should also be noted that the proposed approach does inherit all the limitations associated with the NGM method. For instance, it provides only a local stability condition and does not allow to make a conclusion about the behavior of the system under large deviations. 

\section*{Acknowledgment}
Research presented in this article was supported by the Laboratory Directed Research and Development program of Los Alamos National Laboratory under project number 20180612ECR.

\section*{Appendix A. Proofs and computations}\label{sec:appA}
\paragraph{Proof of Lemma \ref{lem:linearize-u}} One can readily observe that $\dfrac{\partial \Phi^{u}_{i,k}(x)}{\partial x_j}=0$ for all $j>l$ as $\Phi^{u}_{i,k}(x^*)$ turns to zero identically for any $x^*=\begin{bmatrix}\mathbf{0}_l& x^C\end{bmatrix}^\top$. We thus consider the partial derivatives of $\Phi^{u}_{i,k}(x^*)$ w.r.t. $x_j$ for $j=1,\dots,l$. 

Let $i,j\in\{1,\dots,l\}$. We have
\begin{align*}\dfrac{\partial \Phi^{u}_{i,k}(x)}{\partial x_j}
=&{ }\lim_{\delta x_j\rightarrow 0}\dfrac{\sum\limits_{\substack{q\le l\\q\neq i}}a^u_{iq,k}(0,\dots,\delta x_j,\dots,0)-\sum\limits_{q\neq i}a^u_{qi,k}(0,\dots,\delta x_j,\dots,0)}{\delta x_j}\\[10pt]
=&{ }\begin{cases}\lim\limits_{\delta x_j\rightarrow 0}\dfrac{-\sum\limits_{q\neq i}a^u_{qi,k}(0,\dots,\delta x_j,\dots,0)}{\delta x_j}\le 0,&i=j\\[10pt]
\lim\limits_{\delta x_j\rightarrow 0}\dfrac{a^u_{ij,k}(0,\dots,\delta x_j,\dots,0)}{\delta x_j}\ge 0&i\neq j \end{cases}\end{align*}
For $i\in\{l+1,\dots,n\}$ and $j\in\{1,\dots,l\}$, the partial derivatives are
\begin{align*}\dfrac{\partial \Phi^{u}_{i,k}(x)}{\partial x_j}
=&{ }\lim_{\delta x_j\rightarrow 0}\dfrac{a^u_{ij,k}(0,\dots,\delta x_j,\dots,0)}{\delta x_j}\ge 0,
\end{align*}
which yields the required sign structure. Summation over $j$ gives $\sum_{j=1}^n \dfrac{\partial \Phi^{u}_{i,k}(x)}{\partial x_j} =0$ for all $i=1,\dots,l$ and hence, we have (note that the summation is performed only for $j\le l$) \begin{equation}\label{eq:diag-dom}\left|\dfrac{\partial \Phi^{u}_{i,k}(x)}{\partial x_j}\right|\ge \sum\limits_{\substack{j\le l\\j\neq i}} \left|\dfrac{\partial \Phi^{u}_{i,k}(x)}{\partial x_j}\right|.\end{equation}
This implies that the matrix $B_{11,k}=\left[\dfrac{\partial \Phi^{u}_{i,k}(x)}{\partial x_j}\right]_{\substack{i=1,\dots,l\\j=1,\dots,l}}$ is weakly column diagonally dominant as the inequality in (\ref{eq:diag-dom}) is not strict. \hfill$\Box$
\paragraph{Computation of the second order term in the expansion of $R_0(u)$}
Let $A(u)$ be a matrix depending on $u$, $r_0$ be the simple eigenvalue equal to the spectral radius of $A(0)$, and $w_0$ and $v_0$ be the left and the right eigenvectors corresponding to $r_0$.
We expand the perturbed eigenvalue $r(u)$ and the corresponding right eigenvector $x(u)$ in a Taylor series and keep the terms up to the second order: $r(u)=r_0+r_1 u + r_2 u^2 + O(u^3)$ and $x(u)=v_0+v_1 u + v_2 u^2 + O(u^3)$. Thus we have
\begin{equation}\label{eq:Au-1}A(u)(v_0+v_1 u + v_2 u^2 + O(u^3))=(v_0+v_1 u + v_2 u^2 + O(u^3))(r_0+r_1 u + r_2 u^2 + O(u^3))\end{equation}
Differentiating (\ref{eq:Au-1}) w.r.t. $u$ and evaluating at $u=0$ we get
\begin{equation}\label{eq:Au-2}(I r_0-A(0))v_1=(A'(0) - I r_1)v_0,\end{equation}
whence the expression for $r_1$ can be obtained: $r_1 =w_0^\top A'(0)v_0(w_0^\top v_0)^{-1}$ (cf.\ the proof of Thm.\ \ref{thm:R1}). Substituting $r_1$ back to (\ref{eq:Au-2}) one gets an expression that can be used to determine $v_1$ (see \cite[Chap. 8]{Magnus:99} for details):
\begin{equation}\label{eq:Au-3}v_1=(I r_0-A(0))^\dagger\left(I - \frac{v_0 w_0^\top}{w_0^\top v_0}\right) A'(0)v_0,\end{equation}
where $()^\dagger$ is the Moor-Penrose inverse operator.  

To compute the second term in the expansion of $r(u)$ we differentiate (\ref{eq:Au-1}) twice w.r.t. $u$ and evaluate at $u=0$ to get
$$\frac{1}{2}A''(0)v_0 +  A'(0) v_1 - v_1 r_1 - v_0 r_2=   (I r_0 - A(0)) v_2 $$
Multiplying from the left by $w_0^\top$ and substituting the previously obtained expressions for $r_1$ and $v_1$ we arrive after some computations to the final expression for $r_2$:
\begin{equation}\label{eq:r2}r_2=\left(w_0^\top v_0\right)^{-1}w_0^\top \left[\frac{1}{2} A''(0) + A'(0)P_0(I r_0-A(0))^\dagger P_0 A'(0)\right]v_0,\end{equation}
where $P_0=I - \frac{v_0 w_0^\top}{w_0^\top v_0}$ is the oblique projection operator. 
Finally, we recall that $A(u)=(V+uW)^{-1}F$, whence $A'(0)=V^{-1}WV^{-1}F$ and $A''(0)=-2V^{-2}WV^{-2}F$.
\section*{Appendix B. Special classes of matrices and their properties.}\label{sec:appB}

\begin{defn}A matrix $B\in \R^{n\times n}$ is {\em non-negative}, denoted by $B\succeq 0$, if $b_{ij}\ge 0$ for all $i,j=1,\dots,n$. 
\end{defn}
\begin{defn}
A non-singular matrix $A$ is said to be {\em inverse positive} if it satisfies $\R^n_{\ge 0}\subseteq A\R^n_{\ge 0}$, which is equivalent to $A^{-1}\succeq 0$, i.e., $A^{-1}\R^n_{\ge 0}\subseteq \R^n_{\ge 0}$. 
\end{defn}
\begin{thm}\label{thm:inv-pos}
Let $A=\alpha I - B$, where $\alpha >0$ and $B\succeq 0$. Then the following statements are equivalent:
\begin{enumerate}
\item The matrix $A$ is inverse positive,
\item The spectral radius of $B$ is strictly smaller than $\alpha$,
\item The matrix $A$ is positive stable, i.e., if $\lambda$ is an eigenvalue of $A$, then $\Re(\lambda)>0$.
\end{enumerate}
\end{thm}
\begin{defn}
A matrix $A=\alpha I - B$ satisfying any of the properties of Thm.\ \ref{thm:inv-pos} is said to be an {\em $M$-matrix}.
\end{defn}
We will occasionally write $\mathbf{M}$ to denote the class of all $n$-by-$n$ non-singular $M$-matrices. The preceding results can be generalized in the following way. Let us define the class of {\em $Z$-matrices} as $\mathbf{Z}=\{A\in \R^{n\times n}|a_{ij}\le 0, i\neq j\}$. The following theorem gives a number of conditions which guarantee that a given Z-matrix is a non-singular M-matrix.  For a complete list see \cite{BP:94}.
\begin{thm}\label{thm:M-charact}
Let $A\in \mathbf{Z}$. Any of the following conditions implies $A\in\mathbf{M}$.
\begin{enumerate}
\item $A$ is inverse-positive.
\item $A$ is positive stable.
\item $A$ has all positive diagonal elements and is {\em strictly row diagonally dominant (d.d.)}, i.e., $$a_{ii}> \sum_{i\neq j}|a_{ij}|,\quad i=1,\dots,n.$$
\item $A$ has all positive diagonal elements and is {\em strictly column d.d.}, i.e., $$a_{ii}> \sum_{i\neq j}|a_{ji}|,\quad i=1,\dots,n.$$
\end{enumerate}
\end{thm}
\begin{proof}
We will prove only the last item as the remaining ones are covered in \cite{BP:94}. 

Let $A\in \mathbf{Z}$ and $A$ be strictly column d.d., then $A^\top\in\mathbf{Z}$ and is strictly row d.d. This implies $A^\top\in \mathbf{M}$. Since the spectrum of $A$ coincides with that of $A^\top$, the positive stability property holds for $A$ and hence $A\in \mathbf{M}$.
\end{proof}



\begin{thebibliography}{10}

\bibitem{Anlar:02}
B~Anlar, A~Ayhan, H~Hotta, M~Itoh, D~Engin, S~Barun, and {\"O}~K{\"o}seoglu.
\newblock Measles virus {RNA} in tonsils of asymptomatic children.
\newblock {\em Journal of paediatrics and child health}, 38(4):424--425, 2002.

\bibitem{Baggaley:05}
Rebecca~F. Baggaley, Neil~M. Ferguson, and Geoff~P. Garnett.
\newblock The epidemiological impact of antiretroviral use predicted by
  mathematical models: a review.
\newblock {\em Emerging Themes in Epidemiology}, 2(9), 2005.

\bibitem{BP:94}
Abraham Berman and Robert~J. Plemmons.
\newblock {\em Nonnegative matrices in the mathematical sciences}.
\newblock SIAM, 1994.

\bibitem{blower_forecasting_2003}
S.~Blower, E.~J. Schwartz, and J.~Mills.
\newblock Forecasting the future of {HIV} epidemics: the impact of
  antiretroviral therapies \& imperfect vaccines.
\newblock {\em {AIDS} reviews}, 5(2):113--125, 2003.

\bibitem{bootsma_effect_2007}
M.~C.~J. Bootsma and N.~M. Ferguson.
\newblock The effect of public health measures on the 1918 influenza pandemic
  in {U.S.} cities.
\newblock {\em Proceedings of the National Academy of Sciences},
  104(18):7588--7593, 2007.

\bibitem{Bowden:00}
Francis~J Bowden and Geoffrey~P Garnett.
\newblock Trichomonas vaginalis epidemiology: parameterising and analysing a
  model of treatment interventions.
\newblock {\em Sexually Transmitted Infections}, 76(4):248--256, 2000.

\bibitem{Castillo:97}
Carlos Castillo-Chavez and Zhilan Feng.
\newblock To treat or not to treat: the case of tuberculosis.
\newblock {\em Journal of mathematical biology}, 35(6):629--656, 1997.

\bibitem{Castillo:02}
Carlos Castillo-Chavez, Zhilan Feng, and Wenzhang Huang.
\newblock On the computation of {$R_0$} and its role on global stability.
\newblock In D.~Bies, S.~Blower, P.~van~den Driessche, D.~Kirschner, and A.-A.
  Yakubu, editors, {\em Mathematical approaches for emerging and reemerging
  infectious diseases: an introduction}, volume 125 of {\em The IMA Volumes in
  Mathematics and its Applications}, pages 229--250. Springer, 2002.

\bibitem{Coddington:55}
Earl~A. Coddington and Norman Levinson.
\newblock {\em Theory of ordinary differential equations}.
\newblock Tata McGraw-Hill Education, 1955.

\bibitem{Diekmann:12}
Odo Diekmann, Hans Heesterbeek, and Tom Britton.
\newblock {\em Mathematical tools for understanding infectious disease
  dynamics}.
\newblock Princeton University Press, 2012.

\bibitem{Diekmann:00}
Odo Diekmann and Johan Andre~Peter Heesterbeek.
\newblock {\em Mathematical epidemiology of infectious diseases: model
  building, analysis and interpretation}, volume~5 of {\em Mathematical and
  computational biology}.
\newblock John Wiley \& Sons, 2000.

\bibitem{Diekmann:90}
Odo Diekmann, Johan Andre~Peter Heesterbeek, and Johan~A.J. Metz.
\newblock On the definition and the computation of the basic reproduction ratio
  $r_0$ in models for infectious diseases in heterogeneous populations.
\newblock {\em Journal of mathematical biology}, 28(4):365--382, 1990.

\bibitem{french_syphilis_2007}
Patrick French.
\newblock Syphilis.
\newblock {\em BMJ : British Medical Journal}, 334(7585):143--147, January
  2007.

\bibitem{MMB:17}
Dmitry Gromov, Ingo Bulla, Oana~Silvia Serea, and Ethan~O. Romero-Severson.
\newblock Numerical optimal control for {HIV} prevention with dynamic budget
  allocation.
\newblock {\em Mathematical Medicine and Biology: A Journal of the IMA}, page
  dqx015, 2017.

\bibitem{Heesterbeek:02}
J.A.P. Heesterbeek.
\newblock A brief history of {$R_0$} and a recipe for its calculation.
\newblock {\em Acta Biotheoretica}, 50(3):189--204, Sep 2002.

\bibitem{Heesterbeek:96}
J.A.P. Heesterbeek and K.~Dietz.
\newblock The concept of {$R_0$} in epidemic theory.
\newblock {\em Statistica Nederlandica}, 50(1):89--110, 1996.

\bibitem{Heffernan:05}
J.M. Heffernan, R.J. Smith, and L.M. Wahl.
\newblock Perspectives on the basic reproductive ratio.
\newblock {\em Journal of the Royal Society Interface}, 2(4):281--293, 2005.

\bibitem{HK:15}
Christopher~J Henry and James~S Koopman.
\newblock Strong influence of behavioral dynamics on the ability of testing and
  treating {HIV} to stop transmission.
\newblock {\em Scientific reports}, 5:9467, 2015.

\bibitem{Horn:91}
Roger~A. Horn and Charles~R. Johnson.
\newblock {\em Topics in matrix analysis}.
\newblock Cambridge University Presss, Cambridge, 1991.

\bibitem{Korenromp:02}
Eline~L Korenromp, Mondastri~K Sudaryo, Sake~J de~Vlas, Ronald~H Gray, Nelson~K
  Sewankambo, David Serwadda, Maria~J Wawer, and J~Dik~F Habbema.
\newblock What proportion of episodes of gonorrhoea and chlamydia becomes
  symptomatic?
\newblock {\em International Journal of STD \& AIDS}, 13(2):91--101, 2002.

\bibitem{Korenromp:00}
Eline~L Korenromp, Carina Van~Vliet, Heiner Grosskurth, Awene Gavyole,
  Catharina~PB Van~der Ploeg, Lieve Fransen, Richard~J Hayes, and J~Dik~F
  Habbema.
\newblock Model-based evaluation of single-round mass treatment of sexually
  transmitted diseases for {HIV} control in a rural {A}frican population.
\newblock {\em Aids}, 14(5):573--593, 2000.

\bibitem{Magnus:99}
Jan~R. Magnus and Heinz Neudecker.
\newblock {\em Matrix differential calculus with applications in statistics and
  econometrics}.
\newblock Wiley, 2nd rev. edition, 1999.

\bibitem{Martin:13}
Natasha~K Martin, Peter Vickerman, Jason Grebely, Margaret Hellard, Sharon~J
  Hutchinson, Viviane~D Lima, Graham~R Foster, John~F Dillon, David~J Goldberg,
  Gregory~J Dore, et~al.
\newblock Hepatitis {C} virus treatment for prevention among people who inject
  drugs: modeling treatment scale-up in the age of direct-acting antivirals.
\newblock {\em Hepatology}, 58(5):1598--1609, 2013.

\bibitem{Meyer:00}
Carl~D Meyer.
\newblock {\em Matrix analysis and applied linear algebra}.
\newblock SIAM, 2000.

\bibitem{WHO:05staging}
World~Health Organization et~al.
\newblock Interim who clinical staging of hvi/aids and hiv/aids case
  definitions for surveillance: African region.
\newblock Technical report, Geneva: World Health Organization, 2005.

\bibitem{Sharomi:08}
Oluwaseun Sharomi, Chandra~N. Podder, Abba~B. Gumel, and Baojun Song.
\newblock Mathematical analysis of the transmission dynamics of {HIV}/{TB}
  coinfection in the presence of treatment.
\newblock {\em Mathematical Biosciences \& Engineering}, 5(1):145--174, 2008.

\bibitem{SilvaTorres:15}
Cristiana~J. Silva and Delfim F.~M. Torres.
\newblock A {TB}-{HIV}/{AIDS} coinfection model and optimal control treatment.
\newblock {\em Discrete \& Continuous Dynamical Systems - A}, 35(9):4639--4663,
  2015.

\bibitem{Driessche:02}
Pauline Van~den Driessche and James Watmough.
\newblock Reproduction numbers and sub-threshold endemic equilibria for
  compartmental models of disease transmission.
\newblock {\em Mathematical biosciences}, 180(1):29--48, 2002.

\bibitem{Wald:95}
Anna Wald, Judith Zeh, Stacy Selke, Rhoda~L. Ashley, and Lawrence Corey.
\newblock Virologic characteristics of subclinical and symptomatic genital
  herpes infections.
\newblock {\em New England Journal of Medicine}, 333(12):770--775, 1995.

\end{thebibliography}
\end{document}